\documentclass[journal]{IEEEtran}
\usepackage{amssymb,amsmath}
\usepackage{amsthm}

\newtheorem{theorem}{Theorem}

\newtheorem{remark}{Remark}
\usepackage[usenames, dvipsnames]{color}
\usepackage{cite}
\usepackage{graphicx,subfigure}
\usepackage{psfrag}
\usepackage{url}
\usepackage[latin1]{inputenc}
\usepackage[multiple]{footmisc}
\usepackage[absolute,overlay]{textpos}

\usepackage{pgf}
\usepackage[latin1]{inputenc}
\usepackage{bbm}

\begin{document}

\title{Hybrid Resource Scheduling for Aggregation in Massive Machine-type Communication Networks}

\author{
	\IEEEauthorblockN{Onel L. Alcaraz López, %
		Hirley Alves, 
		Pedro H. J. Nardelli,
		Matti Latva-aho\\
	}
	\thanks{Onel L. Alcaraz López, Hirley Alves and Matti Latva-aho are with the Centre for Wireless Communications (CWC), University of Oulu, Finland.\{onel.alcarazlopez,hirley.alves,matti.latva-aho\}@oulu.fi}
	\thanks{Pedro H. J. Nardelli is with Laboratory of Control Engineering and Digital Systems, Lappeenranta University of Technology, Finland.  pedro.nardelli@lut.fi}
	\thanks{This work is partially supported by Academy of Finland (Aka) (Grants n.303532, n.307492, n.318927 (6Genesis Flagship)), SRC/Aka BCDC Energy (n. 292854),  as well as the Finnish Foundation for Technology Promotion, the Finnish Funding Agency for Technology and Innovation (Tekes), Bittium Wireless, Keysight Technologies Finland, Kyynel, MediaTek Wireless, Nokia Solutions and Networks.}
}

\maketitle

\begin{abstract}
Data aggregation is a promising approach to enable massive machine-type communication (mMTC). Here, we first characterize the aggregation phase where a massive number of machine-type devices transmits to their respective aggregator. By using non-orthogonal multiple access (NOMA), we present a hybrid access scheme where several machine-type devices (MTDs) share the same orthogonal channel. Then, we assess the relaying phase where the aggregatted data is forwarded to the base station. The system performance is investigated in terms of average number of MTDs that are simultaneously served under imperfect successive interference cancellation (SIC) at the aggregator for two scheduling schemes, namely random resource scheduling (RRS) and channel-dependent resource scheduling (CRS), which is then used to assess the performance of data forwarding phase.
\end{abstract}
\section{Introduction}
Machine-type Communication (MTC) is an inherent part of the fifth generation (5G) cellular networks \cite{Atzori.2017,Chen.2014}, covering automatic data generation, exchange, processing and actuation that are the basis of intelligent machine networks. Such MTC networks are growing
at an impressive rate and some predictions are pointing to 20 billion machine-type devices (MTDs) connected to wireless networks in 2023 and beyond \cite{Ericsson.2017}.
Massive Machine-type Communication (mMTC) are envisaged to cope with that large number of, often low-complexity low-power, MTDs that are becoming part of wireless networks \cite{Dawy.2017}. 
A survey on the requirements, technical challenges, and existing work on medium access control (MAC) layer protocols for supporting these new use cases, is presented in \cite{Rajandekar.2015}, while authors describe also the issues related to efficient, scalable, and fair channel access. In fact, different strategies have been proposed to provide more efficient access,  e.g., access class barring \cite{ETSI1}, prioritized random access \cite{Lin.2014}, backoff adjustment scheme \cite{Yang.2012}, delay-estimation based random access \cite{Hossain.2016}, distributed queuing \cite{Laya.2016}, data aggregation \cite{Chen.2014,Shariatmadari.2015}.
Data aggregation consists in MTDs that organize themselves locally to MTC area networks,  then, the area networks connect to the core network through MTC gateways or data aggregators. This alleviates the problem of massive signaling overhead on the architectural side and it is a key solution strategy to collect, process, and communicate data in MTC use cases with static devices, especially if the locations of the devices are known, such as smart utility meters or video surveillance cameras
\cite{nardelli2016maximizing,Dawy.2017,Ramezanipour.2018}. 

In \cite{Boubiche.2018}, authors survey data aggregation strategies in large-scale wireless sensor networks (WSNs), while focusing on the processing challenges behind the large volume of data. In \cite{Kouzayha.2014}, an experimental study using state-of-the-art drive testing equipment is conducted in order to capture and analyze the impact of MTC data aggregation on signaling overhead in cellular networks with focus on static MTDs such as smart meters and monitoring sensors. Authors of \cite{Riker.2014} present a scheme designed to provide data aggregation for heterogeneous and concurrent sets of Constrained Application Protocol \cite{Bormann.2012} (CoAP) data-requests. The problem of energy-optimal routing and multiple-sink aggregation is investigated in \cite{Fitzgerald.2018}, as well as joint aggregation and dissemination of sensor measurement data in MTC edge networks. An aggregation scheme is proposed in \cite{Shariatmadari.2015} for capillary networks connected to the LTE network to improve their communication efficiency. Authors analyze the trade-offs between random access interaction, resource allocation, and communication latency, and reveal that accepting the extra latency for accumulating packets can significantly reduce the random access requests and the required resources for the data transmissions.

Notice that when aggregating a massive number of MTDs, the density of the aggregators, although it is considerably smaller compared to the density of the MTDs, will still be large. Hence, the interference generated by the devices sharing the same resource is not negligible. There is, though, limited literature considering the interference in mMTC with data aggregation and resource scheduling. Authors in \cite{Guo.2017} partially address those issues by considering a multi-cell network scenario, whose key metrics (namely MTD success probability, average number of successful MTDs and probability of successful channel utilization) are investigated for the random resource scheduling (RRS) and channel-dependent resource scheduling (CRS) schemes. 
 
Another technique called non-orthogonal multiple access (NOMA) is seen as a promising technology for the 5G networks to improve the system spectral efficiency while meeting the demand of massive connectivity demanded by certain MTC applications (e.g \cite{Shirvanimoghaddam.2016}). The key idea behind NOMA is to exploit the power domain for multiple access so that multiple users can be multiplexed at different power levels, but at the same time/frequency/code  employing SIC to separate superimposed messages \cite{Saito.2013_2}. The performance of NOMA is evaluated in \cite{Ding.2014,Zhang.2016} by using the stochastic geometry tools. However, the inter-cell interference, which is a pervasive problem in most of the existing wireless networks, is not explicitly considered in \cite{Ding.2014}, as many other works on NOMA. In contrast, authors in \cite{Zhang.2016} do consider the inter-cell interference when evaluating the performance of NOMA on coverage probability and average achievable rate, but on a downlink setup. In \cite{Lopez.2017}, we propose and analyze a hybrid OMA-NOMA scheme for mMTC uplink scenarios by extending the scheduling schemes RRS and CRS initially proposed in \cite{Guo.2017}. Therein we deal with the limited resources and allow up to two MTDs to share the same orthogonal channel while we consider imperfect SIC. We show that even when the hybrid scheme would lead to a less reliable system with greater chances of outages per MTD, due to the additional intra-cluster interference, it can significantly improve the number of simultaneous active MTDs for high access demand scenarios. 

Differently from \cite{Lopez.2017}, in this work we allow the RRS scheme to control the power coefficients of the MTDs sharing the same channel, thus both, RRS and CRS, have the same impact in terms of interference generated on the outside network (network outside of the aggregation zone of interest). Additionally, here we are agnostic of the outside network topology, nonetheless our proposed model captures the interference coming from outside. We also include the evaluation of the relaying phase of the aggregated data to the base stations, while we focus on the average number of simultaneous active MTDs. That allow us to highlight the advantages of our scheme which aims to provide massive connectivity in scenarios with high access demand, which is not covered by usual OMA setups. Although CRS achieves better performance by providing access to the MTDs with best channels, RRS could be more practical since the random pairing could model scenarios where some MTDs have urgency to be served. Results show that failing to efficiently eliminate the intra-cluster interference could reduce significantly the benefits from NOMA while challenging its practical implementation, thus, power control plays a main role in these systems. Finally, we attain approximated, yet accurate, expressions when analyzing the CRS scheme. In contrast to the time-consuming Monte-Carlo simulations, our analytical derivations allow for fast computation.

Next, Section~\ref{system} introduces the system model. Section~\ref{Agg} discusses the RRS and CRS scheduling schemes for the aggregation phase, while Section~\ref{rel} analyses the relaying phase and the overall system performance. Section~\ref{results} presents the numerical results and Section~\ref{conclusions} concludes the paper.
\begin{figure*}[t!]
	\centering
	\subfigure{\includegraphics[width=0.7\textwidth]{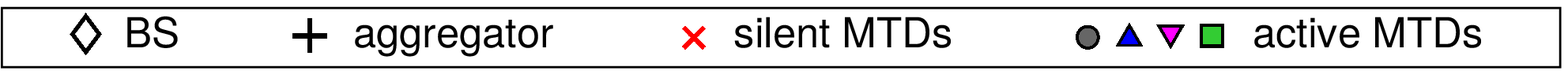}}\\
	\subfigure{\includegraphics[width=0.7\textwidth]{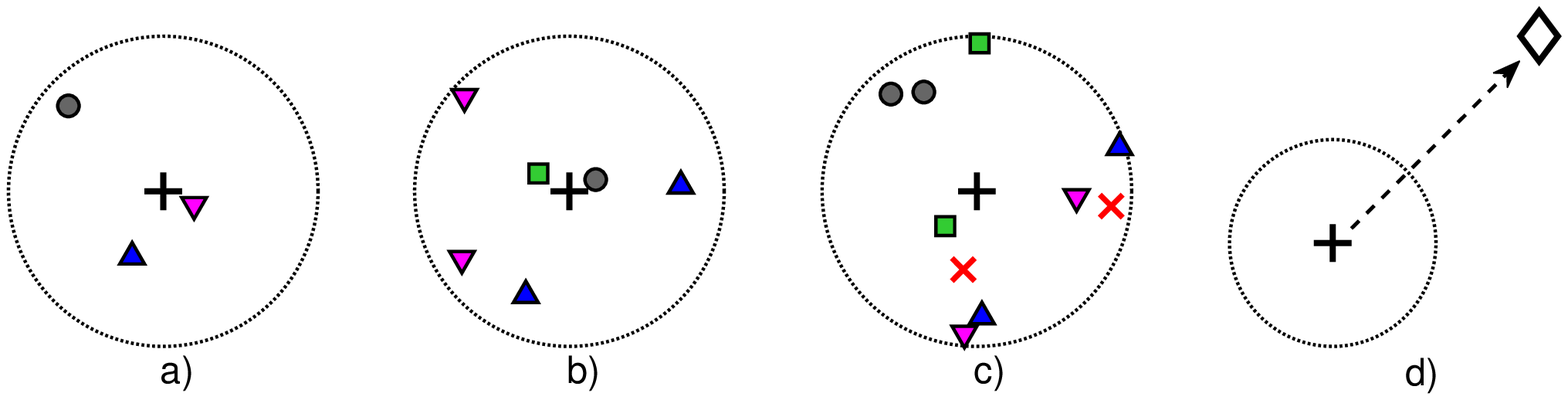}}
	\caption{a), b) and c) Snapshot of the aggregation phase with $\bar{m}=6$, $L=2$ and $N=4$. MTDs with the same shape and color are using the same channel across the entire network. (Different realizations: a) $K<N$, b) $N<K<L\cdot N$, c) $K>L\cdot N$). d) Snapshot of the relaying phase.}		
	\label{model}
\end{figure*}

\textbf{Notation:} $\mathbb{E}[\!\ \cdot\ \!]$ denotes expectation, $\Pr(A)$ is the probability of event $A$, $\Pr(A|B)$ is the $\Pr(A)$ conditioned on $B$. $\mathbbm{1}(\cdot)$ is an indicator function which is equal to $1$ if its argument is true and $0$ otherwise; while $\binom{n}{k}=\frac{n!}{k!(n-k)!}$. $\Gamma(x)$ is the gamma function, $\psi(x)=\tfrac{d[\ln(\Gamma(x))]}{dx}$ is the digamma function,  $Q(a,x)=\tfrac{1}{\Gamma(a)}\int_{x}^{\infty}t^{a-1}e^{-t}\mathrm{d}t$ is the  regularized incomplete gamma function, and $\ _1F_1(a;b;z)$ is the Kummer confluent hypergeometric function. $\mathbf{i}=\sqrt{-1}$ and $\mathrm{Im}\{z\}$ is the imaginary part of $z\in\mathbb{C}$. $f_X(x)$ and $F_X(x)$ are the Probability Density Function (PDF) and Cumulative Distribution Function (CDF) of random variable (RV) $X$, respectively. {$X\sim\mathrm{Exp}(1)$ is an exponential distributed RV with unit mean, e.g., $f_X(x)=e^{-x}$ and $F_X(x)=1-e^{-x}$; while $Y\sim\mathrm{Poiss}(\bar{m})$ is a Poisson distributed RV with mean $\bar{m}$, e.g., $\Pr(Y\!=\!y)\!=\!\tfrac{1}{y!}\bar{m}^ye^{-\bar{m}}$ and $F_Y(y)\!=\!Q(y\!+\!1,\bar{m})$.}
\section{System Model and Assumptions}\label{system}
Consider interference-limited\footnote{The interference from other MTDs is much larger than the white noise in the receivers and, therefore, can be ignored. However, notice that the impact of the noise can be easily incorporated into our analysis.} uplink of MTC network is divided into two phases. In the first phase (aggregation phase), the MTD tries to transmit its data with fixed payload size $b$ (bits) to its serving aggregator. The MTDs are served through $N$ orthogonal channels as in \cite{Guo.2017}; however, here the same orthogonal channel could be used for more than one MTD. When the access demand is not so high, the aggregator will be allocating one MTD per channel. But, when the access demand exceeds the availability of orthogonal channels, some MTDs are allowed to share the same orthogonal channel. This scheme is our proposed hybrid OMA-NOMA multiple access scenario \cite{Lopez.2017}. The number of MTDs requiring service is modeled as $K\sim\mathrm{Poiss}(\bar{m})$. The maximum number of users per orthogonal channel is $L$, where $L=1$ reduces to an OMA scenario, and for simplicity we focus on the $L=2$ setup. Furthermore, the scenario with $L=2$ seems more practical than $L>2$ when we take into consideration the processing
complexity for SIC receivers, especially when SIC error propagation is considered as discussed in \cite{Liu.2016}. Figs.\ref{model}a-c show snapshots of the considered aggregation phase for three different realizations. The silent MTDs are those out of the $N\cdot L$ available resources being used by the active MTDs. The aggregator implements the resource scheduling according to one of the schemes presented in Section~\ref{Agg}, and the MTDs considered are those with granted access since the random access in the network is assumed to be performed\footnote{The resource scheduling schemes require that synchronization procedures, as well as the random access stage, are performed in advance. In fact, the work in \cite{Shirvanimoghaddam.2017} proposes a NOMA scheme allowing the combination of random access and data transmissions phases, where our resource scheduling schemes can be easily incorporated to improve the overall system performance. The details of such implementation are out of the scope of this work.} as in \cite{Guo.2017,Malak.2016,Lopez.2017}.

After aggregating the MTDs' data, the aggregator acts as an ordinary cellular user and relays the entire information to its associated BS in the second phase (relaying phase) as shown in Fig.~\ref{model}d. For the aggregation phase with $L>1$, there is both: outside interference $\rho I_{o,1}$ (i.e. interference from MTDs operating on the same channels but being served by other aggregators), and inside interference (i.e. interference from MTDs within the serving area of the aggregator), which are both RVs. For the relaying phase, let $\rho I_{o,2}$ be the interference at the BS from aggregators in other cells operating with the same channel resources. 
Notice that $I_{o,x}$ relies on the outside network topology, which is assumed unknown, but with a Laplace transform of $I_{o,x}$ in the form of 
\begin{equation}
\mathcal{L}_{I_{o,x}}(s)\!=\!  \exp\left( {\!- \phi_x\Gamma\big(1\!+\!\tfrac{2}{\alpha}\big) \Gamma\big(1\!-\!\tfrac{2}{\alpha}\big) s^{\frac{2}{\alpha}}}\right),\ x=\{1,2\},\label{LIo}
\end{equation}
which is a established result from the stochastic geometry for wireless networks generated as PPP with Rayleigh fading, where $\alpha$ is the path-loss exponent,  while $\phi_x$ accounts for network density, characteristics of the aggregation/relaying areas, and others \cite{Haenggi.2012,Guo.2017}. Also, \eqref{LIo} holds under the assumption of using  statistical full inversion power control \cite{Gharbieh.2017} with parameter $\rho$, as we assume here to  guarantee a uniform user experience while saving valuable energy.
The latter implies that devices control their transmit power such that the average signal power received at the serving aggregator/BS is equal to the predefined value $\rho$. Thus, the  instantaneous received power at the receiver side is $\rho h$, where $h\sim\mathrm{Exp}(1)$ is the channel power gain  under quasi-static Rayleigh fading, and $\rho$ does not impact the performance since we assume the network as interference-limited. 

Notice that the process $\phi_x$ could also be dependent of $\alpha$, as discussed in \cite{Haenggi.2014,Ganti.2016}. In that case, the Laplace transform of the interference for different point processes appear to be merely horizontally shifted versions of each other (in dB) as long as their diversity gain is the same. Thus, scaling the threshold $s$ by this SIR gain factor $\beta$,\footnote{$\beta$ will also depend on $s$, but finding $\beta$ for a fixed $s$ already gives a good approximation \cite{Ganti.2016}.} we get $(\beta s)^{\frac{2}{\alpha}}$, where $\beta^{\frac{2}{\alpha}}$ would be included in $\phi_x$. By properly selecting $\phi_x$, the outside interference for any given topology could be then characterized. Finally, full channel state information (CSI) is assumed at receiver side as in \cite{Guo.2017,Ding.2014,Zhang.2016}.
\section{Aggregation Phase}\label{Agg}
In this section we discuss the RRS and CRS scheduling schemes for our hybrid access protocol.
\subsection{RRS for the Hybrid Access}\label{RRS}
Under the RRS scheme, $N$ out of the $K$ instantaneous MTDs requiring transmissions are independently and randomly chosen and then matched, one-to-one, with the orthogonal channels. If $K\!\le\!N$, all MTDs get channel resources, and even $N\!-\!K$ channels will be unused. Otherwise, if $K\!>\!N$,  the channel allocation is executed again by allowing the remaining MTDs to share channels with the already served MTDs. This process is executed repeatedly until all the MTDs are allocated or the maximum number of MTDs per channel, $L$, is reached for all the channels. The inside interference, coming from the MTDs within the same aggregation zone and sharing the same channel, is faced with SIC. The SIR, $\mathrm{SIR}^r_{j,u}$, of the $j$th MTD being decoded on a typical channel, given the number of MTDs $u$ on the same channel and the RRS scheme, is $\mathrm{SIR}^r_{1,1}=\frac{h}{I_{o,1}}$, while
\begin{align}
\mathrm{SIR}^r_{1,2}&=\frac{a_1\max(h_1,h_2)}{I_{o,1}+a_2\min(h_1,h_2)},\label{SIR12}\\
\mathrm{SIR}^r_{2,2}&=\frac{a_2\min(h_1,h_2)}{I_{o,1}+\mu a_1\max(h_1,h_2)},\label{SIR22}
\end{align}
where $\mu\in[0,1]$ is used to model the impact caused by imperfect SIC \cite{Sun.2016}, while $h_1$ and $h_2$ are the channel power coefficients of both MTDs sharing the channel when $u=2$. 

We can weight the power of coexistent nodes on the same channel through $a_1$ and $a_2$ coefficients. Of course, some kind of feedback from the aggregators would be required after pairing the MTDs\footnote{Since up to 2 MTDs can be scheduled to transmit over the same channel, acquiring and using CSI at the MTD side is not appropriate. Instead, the aggregator should acquire the CSI and use it for the scheduling and for determining the power control coefficients; while finally forwarding such information back to the MTDs.}. By letting $a_1+a_2=\delta$ be a fixed value we impose some kind of total transmission power constraint. This is crucial for NOMA scenarios, and here it is particular important in order to control the interference generated on close aggregators in the outside area\footnote{Notice that power constraints are usually linked to each device individually since they are mostly related to hardware limitations. In fact $a_1,a_2\le\delta$,  therefore, we are implicitly considering also individual power constraints. However, since one channel may be occupied by 2 MTDs, by setting $a_1+a_2=\delta$ we are able of controlling the interference generated on the given channel on close aggregators in the outside area, and even if $\delta\approx 1$ we are making it comparable to the interference that would generate a single MTD if operating alone in that channel.}. Also, $\lim\limits_{I_{o,1}\rightarrow 0}\mathrm{SIR}_{1,2}^r$ is unbounded, but $\lim\limits_{I_{o,1}\rightarrow 0}\mathrm{SIR}_{2,2}^r\le \frac{a_2}{a_1\mu}$ since $\min(h_1,h_2)\le\max(h_1,h_2)$, thus the performance of the second MTDs being decoded on the channel is strongly limited by the SIC imperfection parameter, but by properly selecting $a_1,a_2$ that situation can be relaxed. Consider fixed rate coding scheme where the receiver decodes successfully if the SIR exceeds a threshold $\theta>0$, achieving the information rate of $\log_2(1+\theta)$ [bits/symbol], we state the following theorem.
\begin{theorem}\label{the1}
	The RRS success probability, $p^r_{j,u}$, of the $j$th MTD sharing a typical channel conditioned on $u$ MTDs, is given by
	\begin{align}
	p^r_{1,1}&\!=\!\mathcal{L}_{I_{o,1}}(\theta),\label{p11}\\
	p^r_{1,2}&\!=\!\!\left\{\!\! \begin{array}{ll}\!
	\frac{2a_1}{a_1\!+\!\theta a_2}\mathcal{L}_{I_{o,1}}(\tfrac{\theta}{a_1})\!-\!\frac{a_1\!-\!\theta a_2}{a_1\!+\!\theta a_2}\mathcal{L}_{I_{o,1}}\Big(\!\frac{2\theta}{a_1\!-\!\theta a_2}\!\Big),&\!\!\! \mathrm{if}\ 0\!\le\!\tfrac{\theta a_2}{a_1}\!<\!1\\
	\frac{2a_1}{a_1+\theta a_2}\mathcal{L}_{I_{o,1}}(\tfrac{\theta}{a_1}),&\!\!\! \mathrm{if}\ \tfrac{\theta a_2}{a_1}\ge 1
	\end{array}
	\right.\!\!\!\!\!,\label{p12}		\\
	p^r_{2,2}&=\!\left\{ \begin{array}{ll}
	\frac{a_2-\theta\mu a_1}{a_2+\theta\mu a_1}\mathcal{L}_{I_{o,1}}\Big(\frac{2\theta}{a_2-\theta\mu a_1}\Big),& \mathrm{if}\ 0\le\tfrac{\theta\mu a_1}{a_2}<1\\
	0,& \mathrm{if}\ \tfrac{\theta\mu a_1}{a_2}\ge 1
	\end{array}
	\right.\!\!.\label{p22}
	\end{align}
\end{theorem}
\begin{proof}
	See Appendix~\ref{App_B}.\phantom\qedhere
\end{proof}
\begin{remark}\label{remark1}
	As long as $\theta^2\mu<1$, it is advisable choosing $a_1,a_2$ such that $\theta<\tfrac{a_1}{a_2}<\tfrac{1}{\theta\mu}$, and both MTDs operating on the same channel get the success probability shown in the first line of \eqref{p12} and \eqref{p22}. Notice that by going closer to $\tfrac{1}{\theta \mu}$ we favor the first MTD being decoded, while if we choose a smaller $\tfrac{a_1}{a_2}$, the second MTDs benefits. However, finding the values of $a_1$ and $a_2$ for which the MTDs could perform with similar reliability for any setup, seems intractable. 
\end{remark}

\begin{theorem}\label{the2}
	The Probability Mass Function (PMF) of the number of active MTDs, $K_1^r$, is given in \eqref{cor1} at the top of the next page, where $f_1(k_1^r)=\min(k_1^r,2N-k)$, $f_2(k_1^r,r_1)=\min(k_1^r-r_1,k-N)$ and $f_3(k_1^r)=\min(k_1^r,N)$.
\end{theorem}
\begin{proof}
	See Appendix~\ref{App_B2}.\phantom\qedhere
\end{proof}
\begin{figure*}[!t]
	\footnotesize
	\begin{align}\label{cor1}
	&\Pr(K_1^r\!=\!k_1^r)=
	e^{-\bar{m}p_{1,1}^r}(\bar{m}p_{1,1}^r)^{k_1^r}\Bigg[\frac{1}{k_1^r!}-\frac{(N-k_1^r+1)\binom{N+1}{k_1^r}\Big((N-k_1^r)!-\Gamma\big(N-k_1^r+1,\bar{m}(1-p_{1,1}^r)\big)\Big)}{(N+1)!}\Bigg]+\sum_{k=N+1}^{2N-1}\sum\limits_{r_1=0}^{f_1(k_1^r)}\!\sum\limits_{r_2=0}^{f_2(k_1^r,r_1)}\!\Bigg[ \nonumber\\
	&\qquad\qquad\qquad\mathbbm{1}\big(k_1^r\!\le\! k\!-\!N\!+\!r_1\!+\!r_2\!\big)\!\binom{2N\!-\!k}{r_1}\!\binom{k\!-\!N}{r_2}\!\binom{k\!-\!N}{k_1^r\!-\!r_1\!-\!r_2}\!(p_{1,1}^r)^{r_1}\!(p_{1,2}^r)^{r_2}\!(p_{2,2}^r)^{k_1^r\!-\!r_1\!-\!r_2}\!(1\!-\!p_{1,1}^r)^{2N\!-\!k\!-\!r_1}\!(1\!-\!p_{1,2}^r)^{k\!-\!N\!-\!r_2}\cdot\nonumber\\
	&\cdot(1\!-\!p_{2,2}^r)^{k\!-\!N\!-\!k_1^r\!+\!r_1\!+\!r_2}\frac{e^{-\bar{m}}\bar{m}^k}{k!}\Bigg]\!+\!\big(1-Q(2N,\bar{m})\big)\sum\limits_{r_1=0}^{f_3(k_1^r)}\!\!\!\mathbbm{1}\big(k_1^r\!\le\! N\!+\!r_1\big)\binom{N}{r_1}\binom{N}{k_1^r\!-\!r_1}(p_{1,2}^r)^{r_1}(p_{2,2}^r)^{k_1^r\!-\!r_1}(1\!-\!p_{1,2}^r)^{N\!-\!r_1}(1\!-\!p_{2,2}^r)^{N\!-\!k_1^r\!+\!r_1}. 
	\end{align}	
	\hrule
\end{figure*}
\subsection{CRS for the Hybrid Access}\label{CRS}
The CRS scheme seeks to make better use of channel resources by strongly relying on all the CSI available for scheduling. The MTD with better fading (equivalently, better SIR) will be preferentially assigned with the available channel resources. An aggregator with $K$ instantaneous  MTDs requiring transmission has the knowledge of their fading gains. Let $\{h_1,...,h_i,...,h_K\}$ denote the decreasing ordered channel gains, where $h_{i-1}>h_i$. If $K\le N$ all the MTDs will be chosen, but if $K>N$ the aggregator will pick the $N$ MTDs with better channel gains, i.e., $h_1,...,h_N$, and then will assign randomly the orthogonal channels to them \cite{Guo.2017}. As a continuation, the remaining MTDs can  be still allocated sharing those same resources, i.e., users $N+1$,...,$K$ go to the second round for allocation. This process is executed repeatedly until all the MTDs are allocated or the maximum number of MTDs per channel, $L$, is reached \cite{Lopez.2017}.\footnote{Notice that imperfect CSI would not only affect the information decoding procedure under this scheme, but also the resource scheduling stage since the channel coefficients' ordering may be affected. A detailed performance analysis under imperfect CSI is regarded as our future work.} 

Under the CRS scheme and using SIC to face the inside interference, the SIR, $\mathrm{SIR}^{c\ (i)}_{j,u}$, of the $j$th MTD being decoded on a typical channel, given the first MTD allocated there has the $i$th larger channel coefficient, $h_i$, and there are $u$ MTDs sharing that same channel, is given by $\mathrm{SIR}^{c\ (i)}_{1,1}=\frac{h_i}{I_{o,1}}$ and 
\begin{align}
\mathrm{SIR}^{c\ (i)}_{1,2}&=\frac{a_1^{(i)}h_i}{I_{o,1}+a_2^{(i)}h_{i+N}},\label{sir12}
\end{align}
\begin{align}
\mathrm{SIR}^{c\ (i)}_{2,2}&=\frac{a_2^{(i)}h_{i+N}}{I_{o,1}+\mu a_1^{(i)}h_i}.\label{sir22}
\end{align}
Notice that the bound performance is the same as previously discussed for the RRS scheme. 
Meanwhile, the feedback/signaling overhead is also the same as for the RRS scheme since the CSI acquisition would take place at the aggregator side, which in turn will only forward back the channel allocation for each MTD and the power control coefficients if necessary. We have assumed that such metadata information is sufficiently small such that the low-rate feedback is error-free. However, practical performances would be upper bounded by our results.

Now, assuming that the receiver can decode successfully (SIR exceeds a threshold $\theta$), we state the following theorem.
\begin{theorem}\label{the3}
	 Given the first MTD allocated has the $i$th largest channel coefficient, $h_i$, and that $u$ MTDs share that same channel, the CRS success probability, $p^{c\ (i)}_{j,u}$, of the $j$th MTD being decoded on a typical channel is approximated as
		\small
	\begin{align}
	p^{c\ (i)}_{j,\!u}\!&\!\approx\!\frac{1}{2}\!-\!\frac{1}{\pi}\!\!\int\limits_{0}^{\infty}\!\!\frac{e^{\!-\chi\cos(\tfrac{\pi}{\alpha})\varphi^{\frac{2}{\alpha}}}\!\!\sin\!\big(\chi\sin(\tfrac{\pi}{\alpha})\varphi^{\frac{2}{\alpha}}\!-\!\varphi B_{j,u}^{(i,K)}\!\big)}{\varphi}\mathrm{d}\varphi,\label{pc}
	\end{align}
	\normalsize
	where $\chi=\phi_1\Gamma\big(1+\tfrac{2}{\alpha}\big)\Gamma\big(1-\tfrac{2}{\alpha}\big)$, and
	\small
	\begin{align}
	B_{1,1}^{(i,K\!)}&\!\!\!=\!\!\frac{\psi(K+1)-\psi(i)}{\theta},\label{B11}\\
	B_{1,2}^{(i,K\!)}&\!\!\!=\!\!\Big(\!\tfrac{a_1^{(i)}}{\theta}\!-\!a_2^{(i)}\!\Big)\psi(K\!+\!1)\!+\!a_2^{(i)}\psi(i\!+\!N)\!-\!\tfrac{a_1^{(i)}}{\theta}\psi(i),\label{B12}\\
	B_{2,2}^{(i,K\!)}&\!\!\!=\!\!\Big(\!\tfrac{a_2^{(i)}}{\theta}\!-\!\mu a_1^{(i)}\!\Big)\!\psi(K\!+\!1)\!+\!\mu a_1^{(i)}\psi(i)\!-\!\tfrac{a_2^{(i)}}{\theta}\psi(i\!+\!N).\label{B22}
	\end{align}
	\normalsize
\end{theorem}
\begin{proof}
Theorem~3 in \cite{Lopez.2017} states that 
\begin{align}
p^{c\ \! (i)}_{j,u}\!\approx\!\frac{1}{2}\!-\!\frac{1}{\pi}\!\int\limits_{0}^{\infty}\!\frac{1}{\varphi}\mathrm{Im}\bigg\{\!\mathcal{L}_{I_{o,1}}(\!-\mathbf{i}\varphi)e^{\!-\!\mathbf{i}\varphi B^{i,K}_{j,u}}\!\bigg\}\mathrm{d}\varphi.\label{pcApD}
\end{align}
	Substituting \eqref{LIo} into \eqref{pcApD} along with some algebraic transformations, e.g., $\chi=\phi_1\Gamma\big(1+\tfrac{2}{\alpha}\big)\Gamma\big(1-\tfrac{2}{\alpha}\big)$, $(-\mathbf{i})^{\frac{2}{\alpha}}=\cos(\tfrac{\pi}{\alpha})-\mathbf{i}\sin(\tfrac{\pi}{\alpha})$ and $\mathrm{Im}\{pe^{-q\mathbf{i}}\}=-p\sin(q)$, renders \eqref{pc}.
\end{proof}

\begin{figure*}[!t]
	\footnotesize
	\begin{align}\label{cor3}
	&\Pr(K_1^c\!=\!k_1^c)\approx\nonumber\\
	&
	e^{-\bar{m}p_{1,1}^r}(\bar{m}p_{1,1}^r)^{k_1^r}\Bigg[\frac{1}{k_1^r!}-\frac{(N-k_1^r+1)\binom{N+1}{k_1^r}\Big((N-k_1^r)!-\Gamma\big(N-k_1^r+1,\bar{m}(1-p_{1,1}^r)\big)\Big)}{(N+1)!}\Bigg]\!+\!\!\sum_{k=N+1}^{2N-1}\sum\limits_{r_1=0}^{f_1(k_1^r)}\!\sum\limits_{r_2=0}^{f_2(k_1^r,r_1)}\!\Bigg[\!\mathbbm{1}\big(k_1^c\!\le\! k\!-\!N\!+\!r_1\!+\!r_2\big)\cdot\! \nonumber\\
	&\!\cdot\binom{2N\!-\!k}{r_1}\!\binom{k\!-\!N}{r_2}\!\binom{k\!-\!N}{k_1^c\!-\!r_1\!-\!r_2}\!\big(\bar{p}_{1,1}^c(k)\big)^{r_1}\!\big(\bar{p}_{1,2}^{c_1}(k)\big)^{r_2}\!\big(\bar{p}_{2,2}^{c_1}(k)\big)^{k_1^c\!-\!r_1\!-\!r_2}\!\big(1\!-\!\bar{p}_{1,1}^c(k)\big)^{2N\!-\!k\!-\!r_1}\!\big(1\!-\!\bar{p}_{1,2}^{c_1}(k)\big)^{k\!-\!N\!-\!r_2}\!\big(1\!-\!\bar{p}_{2,2}^{c_1}(k)\big)^{k\!-\!N\!-\!k_1^c\!+\!r_1\!+\!r_2}\cdot\nonumber\\ 
	&\qquad\cdot\frac{e^{-\bar{m}}\bar{m}^k}{k!}\Bigg]\!+\!\sum_{k=2N}^{\infty}\sum\limits_{r_1=0}^{f_3(k_1^c)}\!\!\!\mathbbm{1}\big(f_5(k_1^c)\!\le\! N\big)\binom{N}{r_1}\binom{N}{f_5(k_1^c)}\big(\bar{p}_{1,2}^{c_2}(k)\big)^{r_1}\big(\bar{p}_{2,2}^{c_2}(k)\big)^{f_5(k_1^c)}\big(1\!-\!\bar{p}_{1,2}^{c_2}(k)\big)^{N\!-\!r_1}\big(1\!-\!\bar{p}_{2,2}^{c_2}(k)\big)^{N\!-\!f_5(k_1^c)}\frac{e^{-\bar{m}}\bar{m}^k}{k!}. 
	\end{align}	
	\hrule
\end{figure*}

\begin{theorem}\label{the4}
	The PMF of the number of active MTDs, $K_1^c$, is approximated by \eqref{cor3} at the top of the next page and below \eqref{cor1}, where $f_z(\cdot)$ for $z=1,2,3$ are given in Theorem~\ref{the2} and
	\begin{align}
	\bar{p}_{1,1}^c(k)&=\frac{1}{2N-k}\sum_{i=k-N+1}^{N}p_{1,1}^{c\ (i)},\label{p11a}\\
	\bar{p}_{j,2}^{c_1}(k)&=\frac{1}{k-N}\sum_{i=1}^{k-N}p_{j,2}^{c\ (i)},\label{pj2c1a}\\
	\bar{p}_{j,2}^{c_2}(k)&=\frac{1}{N}\sum_{i=1}^{N}p_{j,2}^{c\ (i)}.\label{pj2c2a}
	\end{align}	
\end{theorem}
\begin{proof}
	The fact that the success probabilities for the CRS scheme, $p_{j,u}^{c\ (i)}$, depend on $i$ and $k$ complicates heavily the problem. Finding their average with regard the index $i$ allows us to use the same procedure when deriving Theorem~\ref{the2} while attaining an accurate approximation. Now when $N\!<\!K\!<\!2N$ in \eqref{K1}, we substitute each probability value, $p_{1,1}^r$ and $p_{j,2}^r$, respectively by \eqref{p11a} and \eqref{pj2c1a} since the success probability depends on the number of MTDs requiring transmission. When $K\ge 2N$ we do similar by replacing $p_{j,2}^r$ by \eqref{pj2c2a} in \eqref{K1}.
\end{proof}
\subsection{Optimum Scheduling: Is it Feasible?}\label{OPT}
Notice that previous scheduling schemes do not guarantee an optimum performance. This is obvious for the case of RRS since such scheme relies entirely on random pairing, while CRS, even when it exploits CSI for making the pairing decisions, is also sub-optimal. As an example, notice that for $K>2N$ a better scheduling when $\mu=0$ and $I_{o,1}\approx 0$ will probably be the  one pairing the $N$ MTDs with better fading conditions with the ones having the worst fading. This is because such pairing benefits always the MTD to be decoded first, while the MTD to be decoded second is not going to be affected by significant outside interference neither by residual interference from SIC.  Consequently, it is expected that as $\mu$ and/or $I_{o,1}$ take meaningful values, the CRS's performance approaches (but not necessary reaches) the optimum.

The optimum scheduling requires an exhaustive search over all the feasible scheduling outcomes in order to adopt the one offering maximum performance. Notice that the dimension of the search space, $D_{K,N}$, depends on $K$ and $N$ since
\begin{itemize}
	\item If $K\le N$ there is only one feasible allocation, which is granting individual channel resources to all MTDs;
	\item if $N<K\le 2N$, there are $2N-K$ MTDs that will be scheduled alone in their channels. Thus, there is a total of $\binom{K}{2N-K}$ for making such selection, while the remaining $2(K-N)$ MTDs need to be paired between each others to share $K-N$ channels; which can be performed in $\big(2(K-N)-1\big)(K-N)$ different ways; 
	\item if $K>2N$, it is necessary selecting the $2N$ MTDs that will get the channel resources for which there are $\binom{K}{2N}$ possibilities; and also making the pairing by testing all the $N(2N-1)$ different alternatives.
\end{itemize}
Therefore,

\begin{align}
D_{K,N}\!=\!\left\{ \begin{array}{ll}
\!\!\!\!1,&\!\!\! \mathrm{if}\ K\le N\\
\!\!\!\!\binom{K}{2N\!-\!K}\big(2(K\!-\!N)\!-\!1\big)(K\!-\!N),\!&\!\!\! \mathrm{if}\ N\!<\!K\!\le\! 2N\\
\!\!\!\!\binom{K}{2N}N(2N-1), &\!\!\! \mathrm{if}\ K\!>\!2N
\end{array}
\right.,\label{DKN}
\end{align}
while on average the dimension of the search space is $\bar{D}_{N}=\sum_{k=0}^{\infty}D_{k,N}\Pr(K=k)$, which can be stated as in \eqref{DN} at the top of the next page. Notice that $(a)$ came from using \eqref{DKN}, while $(b)$ followed from using the CDF of $K$, evaluating the sums and using the definition of the Kummer confluent hypergeometric function, and performing some algebraic transformations and simplifications.
\begin{figure}[t!]
	\centering
	\subfigure{\includegraphics[width=0.42\textwidth]{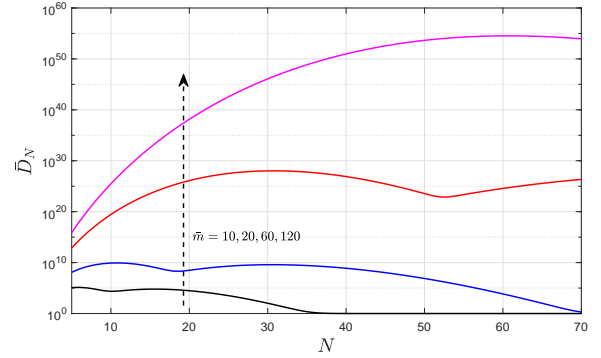}}
	\caption{Average dimension of the search space $\bar{D}_N$ as a function of $N$ for $\bar{m}\in\{10,20,60,120\}$.}		
	\label{Fig2}
\end{figure}
\begin{figure*}[t!]
		\footnotesize
\begin{align}
\bar{D}_{N}&\stackrel{(a)}{=}\Pr(K\le N)+\sum_{k=N+1}^{2N}\binom{k}{2N-k}\big(2(k-N)-1\big)(k-N)\Pr(K=k)+N(2N-1)\sum_{k=2N+1}^{\infty}\binom{k}{2N}\Pr(K=k)\nonumber\\
&\stackrel{(b)}{=}Q(N+1,\bar{m})+\frac{e^{-\bar{m}}\bar{m}^{1+N}}{6(N-1)!}\bigg[3\ _1F_1\Big(1-N,\frac{3}{2},-\frac{\bar{m}}{4}\Big)+\bar{m}(N-1)\ _1F_1\Big(2-N,\frac{5}{2},-\frac{\bar{m}}{4}\Big)\bigg]+\frac{\big(1-e^{-\bar{m}}\big)\bar{m}^{2N}N(2N-1)}{(2N)!}.\label{DN}
\end{align}
\hrule
\end{figure*}

Fig.~\ref{Fig2} shows $\bar{D}_N$ as a function of $N$ for different values of $\bar{m}$. Notice that unless $\bar{m}\ll N$, the dimension of the search space becomes extremely large on average, which makes the exhaustive search unfeasible. Since the scheduling problem appears exactly when the contrary occur, e.g., when $\bar{m}$ is comparable or greater than $N$ since otherwise $1$ MTD per channel is frequently viable, we can conclude that indeed the optimum scheduling through brute force is unfeasible.

In the following section we discuss the overall system performance after analyzing the relaying phase, in which all collected data is forwarded to the BS.
\section{Relaying Phase \& Overall Performance}\label{rel}
In the relaying phase, the aggregator transmits its aggregated data to the BS\footnote{For simplicity, we assume that each aggregator has no buffer and transmits all its aggregated data in one go as in \cite{Lopez.2017,Guo.2017}.}\footnote{Notice that this transmission occurs over only one BS serving channel, therefore, the aggregation topology is reducing the number of BS channel allocations to the MTC devices in a cluster, from $N$ in the case of no aggregation, down to $1$. The importance of such approach is highlighted in \cite{Chen.2014,Shariatmadari.2015}.}. The aggregated data can be successfully decoded by the BS if SIR meets the following condition $\log_2(1+\mathrm{SIR_{rel}}) \ge \tau K_1$, where $\tau=\tfrac{b}{TW}$ in bits per channel use per MTD (bpcu/MTD) with $T$ being the relaying transmission time and $W$ is the available bandwidth for that transmission. $\mathrm{SIR_{rel}}$ is the SIR of the signal received at the BS.  
Then, we write the relaying success probability conditioned on $K_1$ active aggregated MTDs as
\begin{align}
p_{\mathrm{rel}}(K_1)&\!=\!\Pr\big(\mathrm{SIR_{rel}}\ge 2^{\tau K_1}-1\big)\nonumber\\
&\stackrel{(a)}{=}\mathbb{E}_{I_{o,2}}\Big[\Pr\big(h\ge (2^{\tau K_1}-1)I_{o,2}\big)\Big|I_{o,2}\Big]\nonumber\\
&\stackrel{(b)}{=}\mathbb{E}_{I_{o,2}}\Big[\exp\big(-(2^{\tau K_1}-1)I_{o,2}\big)\Big]\nonumber\\
&\stackrel{(c)}{=}\mathcal{L}_{I_{o,2}}(2^{\tau K_1}\!-\!1)\nonumber\\
&\stackrel{(d)}{=}\!\exp\left({\!- \phi_2\Gamma\big(1\!+\!\tfrac{2}{\alpha}\big) \Gamma\big(1\!-\!\tfrac{2}{\alpha}\big) \big(2^{\tau K_1}\!-\!1\big)^{\frac{2}{\alpha}}}\right),\label{pK1}
\end{align}
where $(a)$ comes from using $\mathrm{SIR_{rel}}=g/I_{o,2}$ assuming that $h$ denotes the channel power gain of the link between the aggregator and the BS, $(b)$ comes from using $F_H(h)=1-e^{-h}$, $(c)$ follows after using the definition of the Laplace transform, and finally $(d)$ comes from using \eqref{LIo}.

We now are able to evaluate the average number of successful MTDs, which is an overall performance metric embracing both, the aggregation and relaying phases. That metric evaluates the average number of MTDs being served by the aggregator, whose data can be successfully received by the BS. We can formally write this metric as 
\begin{align}
\bar{K}_{\mathrm{a\&r}}&=\mathbb{E}\Big[K_1 \mathbbm{1}\big(\mathrm{SIR_{rel}}>2^{\tau K_1}-1\big)\Big]\nonumber\\
&=\sum_{k_1=1}^{2N}k_1\Pr(K_1=k_1)p_{\mathrm{rel}}(k_1),
\end{align}
where $\Pr(K_1=k_1)$ is given in \eqref{cor1} and \eqref{cor3} for the RRS and CRS schemes, respectively.

\section{Numerical Results}\label{results}
\begin{figure}[t!]
	\centering
	\subfigure{\includegraphics[width=0.42\textwidth]{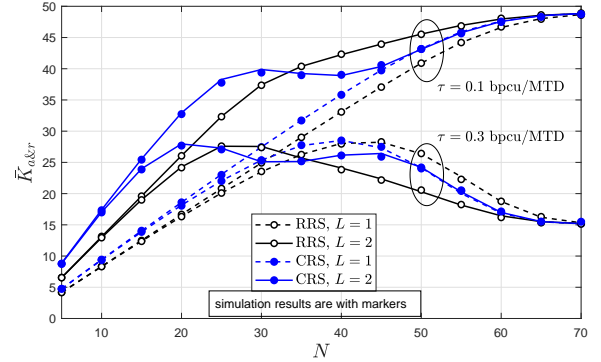}}
	\caption{Average number of successful MTDs as a function of the number of channels for $\tau\in\{0.1,0.3\}$ bpcu/MTD.}		
	\label{figr1}
\end{figure}
Both, simulation and analytical results, are presented in this section to investigate the performance of our hybrid scheme as a function of the system parameters while comparing it with an OMA setup. Unless stated otherwise, results are obtained by setting $\bar{m}=60$, $\alpha=3.6$, $\mu=10\%$, $\theta=1$, $\tau=0.2$ bpcu/MTD and $a_1=a_1^{(i)}=a_2=a_2^{(i)}=\delta/2$. We set $\phi_1=-10$dB, which matches the scenario where all the outside aggregators, serving areas of radius $40$m, are operating with one MTD per channel, while forming a PPP with density $10^{-4.4}/\mathrm{m}^2$. Also, $\delta=1$ such that the average consumed power per orthogonal channel keeps the same for either the $L=1$ or $L=2$ setup, while the interference generated over MTDs sharing the same channel but outside the serving zone keeps similar as in the OMA setup. For the relaying phase we set $\phi_2=-26$dB, which matches the scenario where BSs are serving circular areas of approximately $500$m, while forming a PPP with density $\tfrac{1}{\pi 500}/\mathrm{m}^2$.
Simulation results are generated using 20000 Monte Carlo runs\footnote{Note that  simulations, proposed analytical expressions and approximations fit well in all the cases depicted  Fig.~\ref{figr1}-\ref{figr3}, which validates our findings.}.

Fig.~\ref{figr1} shows that the hybrid scheme for the aggregation phase can improve the spectral efficiency by providing service to a greater number of MTDs when the access demand increases, e.g., $\bar{m}\gtrsim 2N$. This claim comes from \cite{Lopez.2017}, where only the aggregation phase was analyzed. We now extend it by considering the relaying phase, where spite the fact that multiplexing a greater number of MTDs with the same rate degrades the system reliability, the advantage of the hybrid scheme over the purely OMA setup holds. Notice that spectral efficiency for both setups, e.g., $L=1$ and $L=2$, degrades by decreasing $\phi_1$, $\phi_2$ as can be observed in Fig.~\ref{figr2}, and/or by increasing $\tau$. In fact, when $\tau$ increases the degradation of the relaying phase performance due to more data that it is being transmitted could be faster than the increase in the number of active MTDs in the aggregation phase when $N$ increases, thus, the overall performance may worsen as it is shown in Fig.~\ref{figr1} for the case of $\tau=0.3$ bpcu/MTD when $N\gtrsim 45$. On the other hand all the curves tend to overlap when $N$ increases since the probability of having two MTDs sharing the same orthogonal channel decreases so that performance is similar to OMA setup. The fluctuations observed in the CRS scheme for different values of $N$ is consequence of higher dependence/sensitivity on power control coefficients compared to RRS scheme. Thus, a careful selection of those parameters is required for each system setup.
\begin{figure}[t!]
	\centering
	\subfigure{\includegraphics[width=0.42\textwidth]{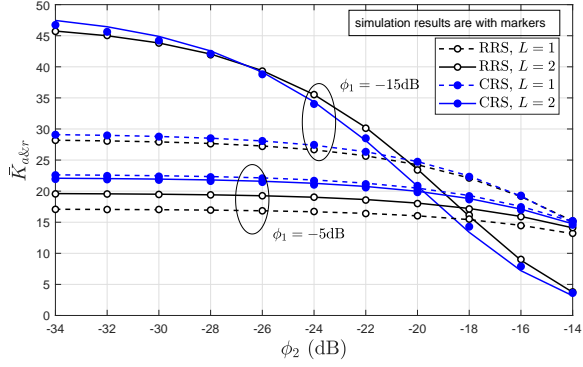}}
	\caption{Relaying phase: average number of successful MTDs as a function of $\phi_2$ for $\phi_1\in\{-15,-5\}$dB and $N=30$.}		
	\label{figr2}
\end{figure}
For the CRS scheme we were able to reach closed-form expressions in \cite[Eqs. (33), (34)]{Lopez.2017} for the power control coefficients while attaining similar reliability for both MTDs sharing the same channel. Notice that when using the RRS scheme, all MTDs have chance to transmit independently of the channel conditions, which could model more realistic scenarios where some MTDs require be served urgently. Also, RRS has a slight decreased performance compared to CRS scheme in the aggregation phase \cite{Lopez.2017}.

Fig.~\ref{fig_n} investigates the required $a_1$ for RRS to attain either similar reliability for both MTDs sharing the same channel or a maximum average number of simultaneously active MTDs in the aggregation phase $\bar{K}$ \cite[Eq. (17)]{Lopez.2017}, as a function of $\phi_1$. As outside interference increases, the required power control coefficient for the first MTD, $a_1$, decreases when similar reliability is the goal, since the performance of the second MTD deteriorates faster and the power control coefficient $a_2$ should increase. Otherwise, when the goal is to maximize the number of simultaneously active MTDs, the performance of the second MTD needs to be sacrificed, even more so when the outside interference increases until the hybrid scheme performs as the OMA setup, e.g., $a_1=1, a_2=0$. Notice that almost all the time, a greater SIC imperfection leads to a reduction in the required $a_1$, decreasing its impact on the performance of the second decoded MTD. Only when reaching $\bar{K}_{\max}$ is the goal and the outside interference is sufficiently large, a  greater SIC imperfection accelerates the transition to OMA by increasing $a_1$. Also, increasing $\phi_1$ deteriorates $\bar{K}$ (see numbered labels in Fig.~\ref{fig_n}).
\begin{figure}[t!]
	\centering
	\subfigure{\includegraphics[width=0.42\textwidth]{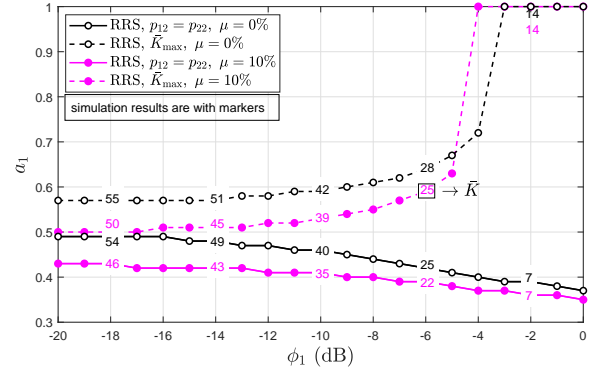}}%
	\caption{$a_1$ ($a_2=\delta-a_1$) as a function of $\phi_1$ for the RRS scheme to attain either similar reliability for both MTDs sharing the same channel or a maximum average number of simultaneously active MTDs in the aggregation phase. $\mu=\{0,10\}\%$ and $N=30$.}	
	\label{fig_n}	
\end{figure}
\begin{figure}[t!]
	\centering
	\subfigure{\includegraphics[width=0.42\textwidth]{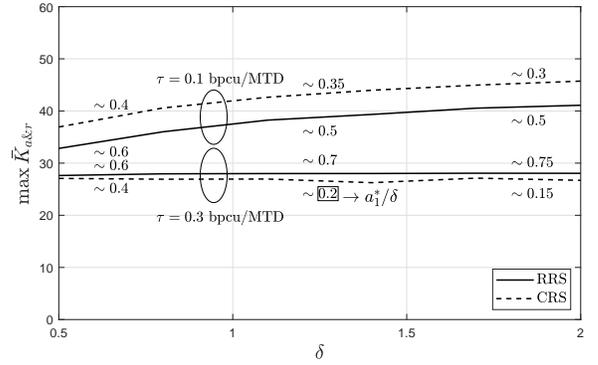}}
	\caption{Maximum average of simultaneously served MTDs as a function of $\delta$, for $N=30$ and $\tau\in\{0.1 0.3\}$ bpcu/MTD.}		
	\label{figr5}
\end{figure}
\begin{figure}[t!]
	\centering
	\subfigure{\includegraphics[width=0.42\textwidth]{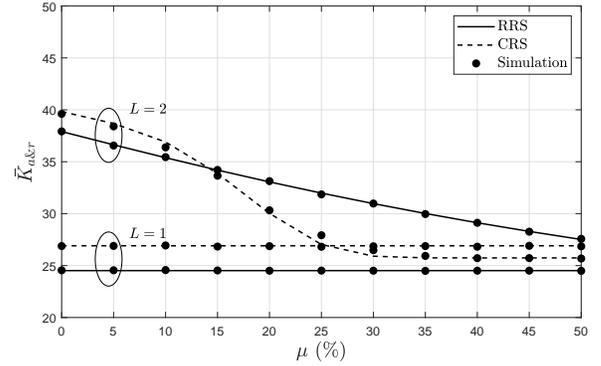}}
	\caption{Average number of successful MTDs as a function of the SIC imperfection parameter for $N=30$.}		
	\label{figr3}
\end{figure}

The reachable maximum average of simultaneously served MTDs as a function total power constraint coefficient, $\delta$, is shown in Fig.~\ref{figr5}. Therein, we find the coefficients $a_1$ and $a_2$ that maximize $\bar{K}_{a\&r}$ such that $a_1+a_2=\delta$. Notice that increasing $\delta$ has a positive impact on the system performance as long as the appropriate values of $a_1$ and $a_2$ are selected, which can be deduced from numbered labels in the figure. For the case of $\tau=0.3$ bpcu/MTD this effect is not evident since the relaying phase is limiting the system performance much more than the aggregation phase. On the other hand, increasing $\delta$ is not always feasible, e.g., due to transmit hardware limitations, or even advisable, e.g., due to the extra interference that might be generated over other OMA networks or because of a low energy efficiency performance. Therefore, the appropriate selection of $\delta$ is of paramount importance.
	
Fig.~\ref{figr3} shows the average number of simultaneously served MTDs as a function of the SIC imperfection coefficient. We set $N=30$ such that each channel is operating with two MTDs almost all the time, which are more sensitive to the interference and imperfection of the SIC. Since SIC is only related with the $L=2$ setup, the OMA setup curves are shown with straight lines.
Of course, when $\mu$ increases, the performance of the $L=2$ setup deteriorates, specifically if the power coefficients are not tuned accordingly. This is because those coefficients work well for certain system parameters but others will be required if they change, e.g. different $\mu$ in this case. It is expected that a smaller $a_1$, hence larger $a_2$, work better as $\mu$ increases as shown previously in Fig.~\ref{fig_n}. It is clear that failing to efficiently  eliminate the inside interference could reduce significantly the benefits from NOMA, and can be a  challenging issue for implementing NOMA in practice.
\vspace{-2mm}
\section{Conclusion}\label{conclusions}
\vspace{-2mm}
We analyzed the data aggregation and relaying in interference-limited mMTC network. We evaluate a hybrid access scheme, OMA-NOMA, while investigating its performance in terms of average number of simultaneously served MTDs. Power control coefficients are incorporated to the practical-interest RRS scheme, while we investigate them numerically. The numerical results also show that our hybrid access scheme aims at providing massive connectivity in scenarios with high access demand. However, inter-cluster interference could reduce significantly the benefits from NOMA, and challenging its implementation in practice. 
\vspace{-2mm}
\appendices 
\vspace{-2mm}
\section{Proof of Theorem~\ref{the1}}\label{App_B}
As in \cite[Th. 1]{Lopez.2017}, let us write the success probabilities as
\begin{align}
p^r_{1,1}&=\mathbb{E}_{I_{o,1}}[\Pr(h>\theta I_{o,1}|I_{o,1})]=\mathbb{E}_{I_{o,1}}\Big[e^{-\theta I_{o,1}}\Big|I_{o,1}\Big],\label{Bp11}\\
p^r_{1,2}\!&=\!\mathbb{E}_{I_{o,1}}\!\Big[\!\Pr\big(\!\max(h_1,h_2)\!-\!\tfrac{\theta a_2}{a_1}\min(h_1,h_2\!)\!\!>\!\tfrac{\theta}{a_1} I_{o,1}|I_{o,1}\big)\!\Big]\nonumber\\
\!&=\!\mathbb{E}_{I_{o,1}}\Big[\Pr\big(v_1>\tfrac{\theta}{a_1} I_{o,1}|I_{o,1}\big)\Big]\nonumber\\
&=1-\mathbb{E}_{I_{o,1}}\Big[F_{V_1}\Big(\tfrac{\theta}{a_1}I_{o,1}\Big)\Big]
\nonumber\\
\!\!\!\!\stackrel{(a)}{=}\!&\!\!\left\{\!\! \begin{array}{ll}
\mathbb{E}\Big[\frac{2a_1}{a_1\!+\!\theta a_2}e^{\!-\!\frac{\theta}{a_1} I_{o,1}}\!-\!\frac{a_1\!-\!\theta a_2}{a_1\!+\!\theta a_2}e^{\!-\!\frac{2\theta}{a_1\!-\!\theta a_2} I_{o,1}}\Big|I_{o,1}\Big]\!,\!& \!\!\mathrm{if} 0\!\le\!\!\tfrac{\theta a_2}{a_1}\!\!<\!1\!\\
\mathbb{E}\Big[\frac{2a_1}{a_1+\theta a_2}e^{\!-\!\frac{\theta}{a_1} I_{o,1}}|I_{o,1}\Big],&\!\!\!\mathrm{otherwise}
\end{array}
\right.\!\label{Bp12}
\end{align}
\begin{align}
p^r_{2,2}\!&=\!\mathbb{E}_{I_{o,1}}\!\Big[\!\Pr\big(\min(h_1,h_2)\!-\!\tfrac{\theta\mu a_1}{a_2}\max(h_1,h_2)\!>\!\tfrac{\theta}{a_2} I_{o,1}|I_{o,1}\!\big)\!\Big]\!\nonumber\\
&=\mathbb{E}_{I_{o,1}}\Big[\Pr\big(v_2>\tfrac{\theta}{a_2} I_{o,1}|I_{o,1}\big)\Big]\nonumber\\
&=1-\mathbb{E}_{I_{o,1}}\Big[F_{V_2}\Big(\tfrac{\theta}{a_2}I_{o,1}\Big)\Big]
\nonumber\\
\label{Bp22}
\!\stackrel{(b)}{=}\!&\!\left\{\!\! \begin{array}{ll}
\mathbb{E}\Big[\frac{a_2-\theta\mu a_1}{a_2+\theta\mu a_1}e^{-\frac{2\theta}{a_2-\theta\mu a_1}I_{o,1}}\Big|I_{o,1}\Big]\!,\!& \mathrm{if}\ 0\!\le\!\tfrac{\theta\mu a_1}{a_2}\!<\!1\!\\
0,&\!\! \mathrm{otherwise}
\end{array}
\right.\!,
\end{align}
where $p^r_{j,u}\!=\!\Pr\big(\mathrm{{SIR}_{j,u}^r}\!>\!\theta\big)$, $v_1=\max(h_1,h_2)-\frac{\theta a_2}{a_1}\min(h_1,h_2)$ and $v_2=\min(h_1,h_2)-\frac{\theta\mu a_1}{a_2}\max(h_1,h_2)$, while $(a)$ and $(b)$ come from using their CDF expressions, which are obtained next.
\begin{align}
F_{V_1}(v_1)&=\Pr(\max(h_1,h_2)-\tfrac{\theta a_2}{a_1}\min(h_1,h_2)\le v_1)\nonumber\\
&=\Pr\big(\max(h_1,h_2)\le v_1+\tfrac{\theta a_2}{a_1}\min(h_1,h_2)\big)\nonumber\\
&=\Pr\big(h_1\le v_1+\tfrac{\theta a_2}{a_1}\min(h_1,h_2)~\bigcap~\nonumber\\
&\qquad\qquad\qquad\qquad h_2\le v_1+\tfrac{\theta a_2}{a_1}\min(h_1,h_2)\big)\nonumber\\
&=\Pr\Big(\min(h_1,h_2)\ge \frac{(h_1-v_1)a_1}{\theta a_2}~\bigcap~\nonumber\\
&\qquad\qquad\qquad\qquad \min(h_1,h_2)\ge \frac{(h_2-v_1)a_1}{\theta a_2}\Big)\nonumber\\
&=\Pr\Big(h_1\!\ge\! \frac{(h_1\!-\!v_1)a_1}{\theta a_2}~\bigcap~h_2\!\ge\! \frac{(h_1\!-\!v_1)a_1}{\theta a_2}~\bigcap\nonumber\\
&\qquad~h_1\ge \frac{(h_2-v_1)a_1}{\theta a_2}~\bigcap~h_2\ge \frac{(h_2-v_1)a_1}{\theta a_2}\Big)\nonumber\\
&=\Pr\Big(h_1(1\!-\!\tfrac{\theta a_2}{a_1})\!\le\! v_1~\bigcap~ h_2\!\ge\! \frac{h_1a_1}{\theta a_2}\!-\!\frac{v_1a_1}{\theta a_2}\bigcap\nonumber\\
&\qquad  h_2\le h_1\tfrac{\theta a_2}{a_1}+v_1~\bigcap~h_2(1-\tfrac{\theta a_2}{a_1})\le v_1\Big). \label{FV1}
\end{align}
\begin{figure}[t!]
	\centering
	\subfigure{\label{Figure4a}\includegraphics[width=0.36\textwidth]{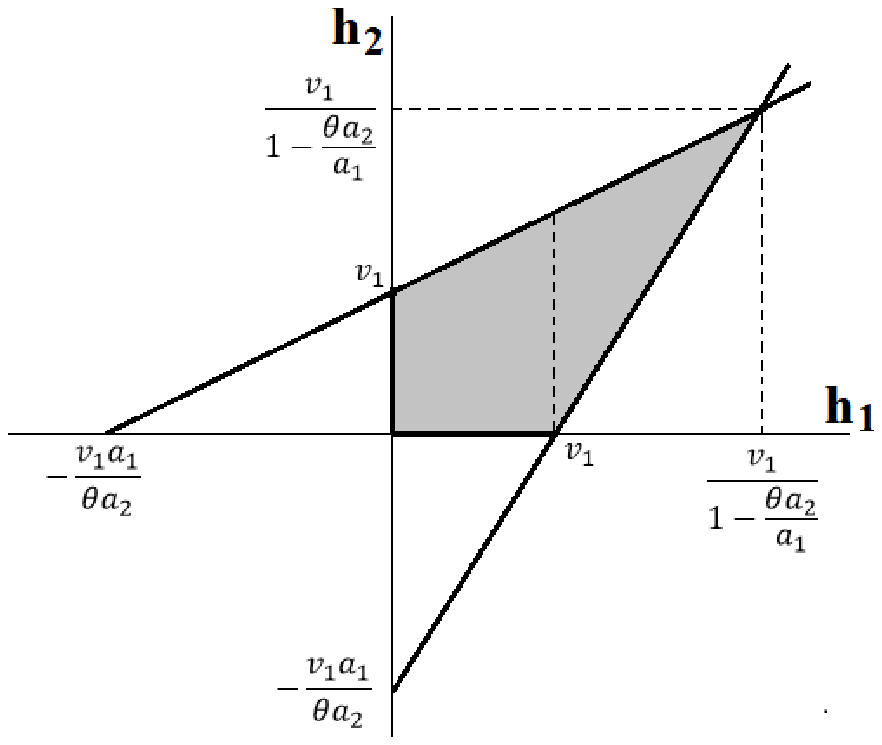}}\ \ \ 
	\subfigure{\label{Figure4b}\includegraphics[width=0.34\textwidth]{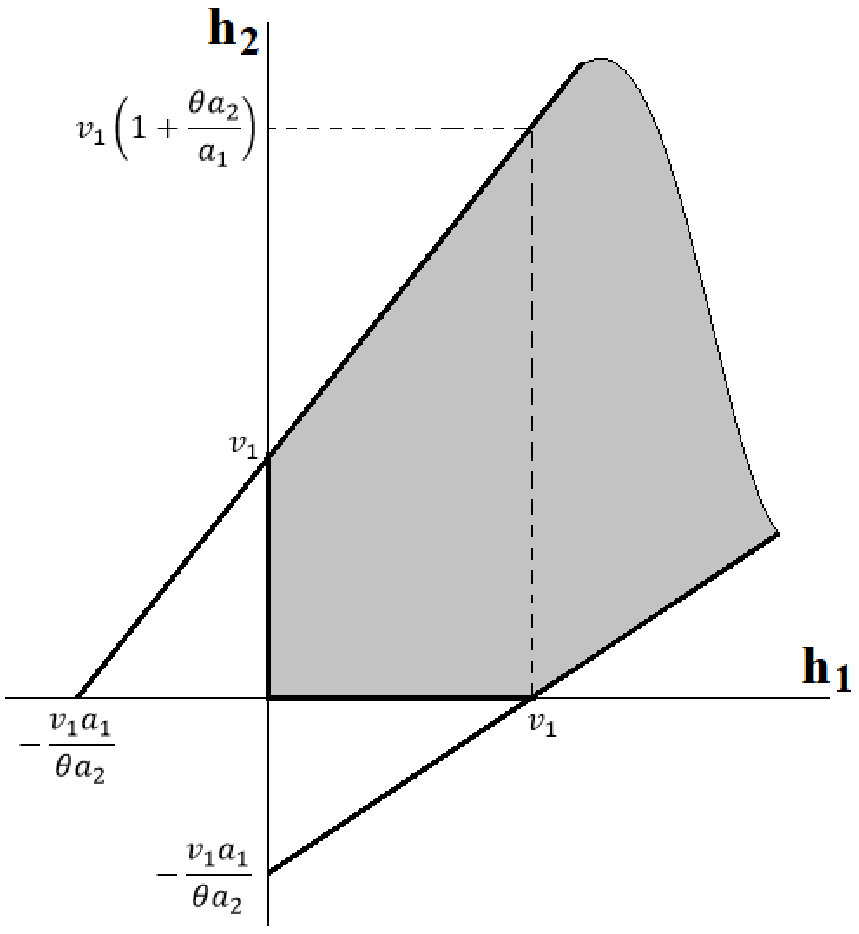}}
	\caption{Region of intersection. a) $0\le\tfrac{\theta a_2}{a_1}<1$ (top), b) $\tfrac{\theta a_2}{a_1}\ge 1$ (bottom).}		
	\label{Figure4}
	\vspace*{-2mm}
\end{figure}

Lets consider the following cases:
\begin{itemize}
	\item $0\le\tfrac{\theta a_2}{a_1}<1$, then we can continue from equation~\eqref{FV1} as follows.
	\begin{align}
	F_{V_1}(v_1)&\!=\!\Pr\Big(h_1\!\le\! \frac{v_1}{1\!-\!\tfrac{\theta a_2}{a_1}}~\bigcap~h_2\!\ge\! \frac{h_1a_1}{\theta a_2}\!-\!\frac{v_1a_1}{\theta a_2}~\bigcap~\nonumber\\
	&\qquad\!\! h_2\le h_1\tfrac{\theta a_2}{a_1}+v_1~\bigcap~h_2\le \frac{v_1}{1-\tfrac{\theta a_2}{a_1}}\Big)  \label{Fv1_c11}
	\end{align}
	and the intersection region is shown in Fig.~\ref{Figure4}a and notice that $h_1,h_2\ge 0$ are also restrictions.
	Therefore we can transform \eqref{Fv1_c11} to attain the result in \eqref{Fv1_c12} at the top of the next page.
	\begin{figure*}[!t]
		\footnotesize
		\begin{align}
	F_{V_1}(v_1)&=\Pr\big(h_1\le v_1~\bigcap~h_2\le \frac{\theta a_2}{a_1} h_1+v_1\big)+\Pr\Big(v_1\le h_1\le \frac{v_1}{1-\frac{\theta a_2}{a_1}}~\bigcap~\frac{a_1}{\theta a_2}h_1-\frac{v_1a_1}{\theta a_2}\le h_2\le \frac{\theta a_2}{a_1} h_1+v_1\Big)\nonumber\\
	&=\int\limits_{0}^{v_1}F_{H_2}(\frac{\theta a_2}{a_1} h_1+v_1)f_{H_1}(h_1)\mathrm{d}h_1+\int\limits_{v_1}^{\frac{v_1}{1-\frac{\theta a_2}{a_1}}}\bigg(F_{H_2}(\frac{\theta a_2}{a_1} h_1+v_1)-F_{H_2}\Big(\frac{h_1a_1}{\theta a_2}-\frac{v_1a_1}{\theta a_2}\Big)\bigg)f_{H_1}(h_1)\mathrm{d}h_1\nonumber\\
	&=\int\limits_{0}^{v_1}\Big(1-e^{-\big(\frac{\theta a_2}{a_1} h_1+v_1\big)}\Big)e^{-h_1}\mathrm{d}h_1+\int\limits_{v_1}^{\frac{v_1}{1-\frac{\theta a_2}{a_1}}}e^{-\big(\frac{h_1a_1}{\theta a_2}-\frac{v_1a_1}{\theta a_2}\big)}e^{-h_1}\mathrm{d}h_1-\int\limits_{v_1}^{\frac{v_1}{1-\frac{\theta a_2}{a_1}}}e^{-\big(\frac{\theta a_2}{a_1} h_1+v_1\big)}e^{-h_1}\mathrm{d}h_1\nonumber\\	
	&=1-\frac{2a_1}{a_1+\theta a_2}e^{-v_1}+\frac{a_1-\theta a_2}{a_1+\theta a_2}e^{-\frac{2v_1a_1}{a_1-\theta a_2}}	 \label{Fv1_c12}
	\end{align}
	\hrule
\end{figure*}
	\item $\tfrac{\theta a_2}{a_1}\ge 1$, then we can continue from equation~\eqref{FV1} as follows.
	\begin{align}
	F_{V_1}(v_1)&\!=\!\Pr\Big(h_1\!\ge \frac{v_1}{1\!-\!\tfrac{\theta a_2}{a_1}}\!~\bigcap~\!h_2\!\ge \frac{h_1a_1}{\theta a_2}-\frac{v_1a_1}{\theta a_2}\!~\bigcap~\nonumber\\
	&\qquad\!\! h_2\le h_1\tfrac{\theta a_2}{a_1}+v_1\!~\bigcap~\!h_2\ge \frac{v_1}{1-\tfrac{\theta a_2}{a_1}}\Big), \label{FV1_c21}
	\end{align}
	Notice that regions $h_1\ge \frac{v_1}{1-\tfrac{\theta a_2}{a_1}}<0$ and $h_2\ge \frac{v_1}{1-\tfrac{\theta a_2}{a_1}}$  do not say nothing new because we already know that $h_1,h_2\ge 0$. Therefore we can write \eqref{FV1_c21} as
	\begin{align}
	F_{V_1}(v_1)&=\Pr\Big(h_2\!\ge\! \frac{h_1a_1}{\theta a_2}\!-\!\frac{v_1a_1}{\theta a_2}\!~\bigcap~\!h_2\!\le\! h_1\frac{\theta a_2}{a_1}\!+\!v_1\Big)\nonumber\\
	&=\Pr\Big( \frac{h_1a_1}{\theta a_2}\!-\!\frac{v_1a_1}{\theta a_2}\le h_2\!\le\! h_1\frac{\theta a_2}{a_1}\!+\!v_1\Big),  \label{FV1_c22}
	\end{align}
	and the intersection region is shown in Fig.~\ref{Figure4}b. Thus, by working on \eqref{FV1_c22} we reach \eqref{FV1_c23} at the top of next page (below \eqref{Fv1_c12}).
	\begin{figure*}[!t]	
			\footnotesize	
	\begin{align}
	F_{V_1}(v_1)&=\int\limits_{0}^{v_1}F_{H_2}\Big(h_1\frac{\theta a_2}{a_1}+v_1\Big)f_{H_1}(h_1)\mathrm{d}h_1+\int\limits_{v_1}^{\infty}\Bigg(F_{H_2}\Big(h_1\frac{\theta a_2}{a_1}+v_1\Big)-F_{H_2}\Big(\frac{h_1a_1}{\theta a_2}-\frac{v_1a_1}{\theta a_2}\Big)\Bigg)f_{H_1}(h_1)\mathrm{d}h_1\nonumber\\
	&=\int\limits_{0}^{\infty}F_{H_2}\Big(h_1\frac{\theta a_2}{a_1}+v_1\Big)f_{H_1}(h_1)\mathrm{d}h_1-\int\limits_{v_1}^{\infty}F_{H_2}\Big(\frac{h_1a_1}{\theta a_2}-\frac{v_1a_1}{\theta a_2}\Big)f_{H_1}(h_1)\mathrm{d}h_1\nonumber\\
	&=\int\limits_{0}^{\infty}\Big(1-e^{-\big(h_1\frac{\theta a_2}{a_1}+v_1\big)}\Big)e^{-h_1}\mathrm{d}h_1-\int\limits_{v_1}^{\infty}\Big(1-e^{-\frac{h_1a_1}{\theta a_2}+\frac{v_1a_1}{\theta a_2}}\Big)e^{-h_1}\mathrm{d}h_1
	=1-\frac{2a_1}{a_1+\theta a_2}e^{-v_1}  \label{FV1_c23}
	\end{align}
	\hrule
\end{figure*}		
\end{itemize}

By combining  \eqref{Fv1_c12} and \eqref{FV1_c23} we attain the general expression for $F_{V_1}(v_1)$, which is
\begin{align}
F_{V_1}(v_1)&\!=\!\!\left\{\!\! \begin{array}{ll}\!\!
1-\frac{2a_1}{a_1\!+\!a_2\theta}e^{-v_1}\!+\!\frac{a_1\!-a_2\!\theta}{a_1\!+\!a_2\theta}e^{-\frac{2a_1v_1}{a_1\!-\!a_2\theta}},&\!\!\mathrm{if}\ 0\le\!\frac{a_2\theta}{a_1}\!<1\\
\\
1-\frac{2a_1}{a_1\!+\!a_2\theta}e^{-v_1},&\!\! \mathrm{if}\ \frac{a_2\theta}{a_1}\!\ge\! 1\\
\end{array}
\right.\!\!\!,
\end{align}
while 
\begin{align}
F_{V_2}(v_2)&\!=\!\left\{ \begin{array}{ll}
1-\frac{a_2\!-\!\theta\mu a_1}{a_2\!+\!\theta\mu a_1}e^{-\frac{2v_2 a_2}{a_2\!+\!\theta\mu a_1}},& \mathrm{if}\ 0\le\tfrac{\theta\mu a_1}{a_2}\!<\!1\\
\\
1,& \mathrm{if}\ \tfrac{\theta\mu a_1}{a_2}\!\ge\! 1\\
\end{array}
\right.\!\!
\end{align}
can be attained by following a similar procedure. 

Finally, by taking the expectation in \eqref{Bp11}, \eqref{Bp12} and \eqref{Bp22} we attain \eqref{p11}, \eqref{p12} and \eqref{p22}, respectively. \hfill\qedsymbol

\begin{figure*}[!t]
	\footnotesize
	\begin{align}	
	\label{K1}
	\Pr(K_1^r\!=\!k_1^r|K)\!=\!\left\{\!\! \begin{array}{ll}\!\!
	0, &K<k_1^r\\
	\!\binom{K}{k_1^r}(p_{1,1}^r)^{k_1^r}(1\!-\!p_{1,1}^r)^{K\!-\!k_1^r}, &k_1^r\!\le\! K\!\le\!N\!\\
	\sum\limits_{r_1=0}^{f_1(k_1^r)}\!\sum\limits_{r_2=0}^{f_2(k_1^r,r_1)}\!\bigg[\mathbbm{1}\big(k_1^r\!\le\! K\!-\!N\!+\!r_1\!+\!r_2\big)\binom{2N\!-\!K}{r_1}\binom{K\!-\!N}{r_2}\binom{K\!-\!N}{k_1^r\!-\!r_1\!-\!r_2}(p_{1,1}^r)^{r_1}(p_{1,2}^r)^{r_2}(p_{2,2}^r)^{k_1^r\!-\!r_1\!-\!r_2}\cdot\\
	\qquad\qquad\qquad\qquad\cdot(1\!-\!p_{1,1}^r)^{2N\!-\!K\!-\!r_1}(1\!-\!p_{1,2}^r)^{K\!-\!N\!-\!r_2}(1\!-\!p_{2,2}^r)^{K\!-\!N\!-\!k_1^r\!+\!r_1\!+\!r_2}\bigg], &N\!<\! K\!<\!2N\!\\ 
	\!\sum\limits_{r_1=0}^{f_3(k_1^r)}\!\!\!\mathbbm{1}\big(k_1^r\!\le\! N\!+\!r_1\big)\binom{N}{r_1}\binom{N}{k_1^r\!-\!r_1}(p_{1,2}^r)^{r_1}(p_{2,2}^r)^{k_1^r\!-\!r_1}(1\!-\!p_{1,2}^r)^{N\!-\!r_1}(1\!-\!p_{2,2}^r)^{N\!-\!k_1^r\!+\!r_1}, &K\!\ge\!2N\!\\ 
	\end{array}
	\right.\!\!
	\end{align}
	\hrule
\end{figure*}
\section{Proof of Theorem~\ref{the2}}\label{App_B2}
For the system model in Section~\ref{system} the outside interference on each channel is independent. Let $K$ be the number of MTDs requiring data transmission, the conditional distribution of $K_1^r$ is given in \eqref{K1} (on the top of the next page) where $f_1(k_1^r)=\min(k_1^r,2N-K)$, $f_2(k_1^r,r_1)=\min(k_1^r-r_1,K-N)$ and $f_3(k_1^r)=\min(k_1^r,N)$. Notice that when $K\le N$ all the MTDs are operating alone in their channels, while when  $N<K<2N$, $2N-K$ MTDs will be operating alone and $K-N$ will be sharing their channels. Additionally, when $K\ge 2N$, the $N$ MTDs with channel resources are sharing their channels. Thus, \eqref{K1} comes from combinatorial and probability theories. The indicator function guarantees operation in the appropriate regions, e.g., $N<K<2N$ and $K\ge N$, according to the particular case. After using $\Pr(K_1^r\!=\!k_1^r)\!=\!\sum_{k=0}^{\infty}\Pr(K_1^r\!=\!k_1^r|k)\Pr(K\!=\!k)$ along with some algebraic manipulations we attain the result in \eqref{cor1}.
\hfill 	\qedsymbol
\bibliographystyle{IEEEtran}
\bibliography{IEEEabrv,references}

\begin{thebibliography}{10}
\providecommand{\url}[1]{#1}
\csname url@samestyle\endcsname
\providecommand{\newblock}{\relax}
\providecommand{\bibinfo}[2]{#2}
\providecommand{\BIBentrySTDinterwordspacing}{\spaceskip=0pt\relax}
\providecommand{\BIBentryALTinterwordstretchfactor}{4}
\providecommand{\BIBentryALTinterwordspacing}{\spaceskip=\fontdimen2\font plus
\BIBentryALTinterwordstretchfactor\fontdimen3\font minus
  \fontdimen4\font\relax}
\providecommand{\BIBforeignlanguage}[2]{{%
\expandafter\ifx\csname l@#1\endcsname\relax
\typeout{** WARNING: IEEEtran.bst: No hyphenation pattern has been}%
\typeout{** loaded for the language `#1'. Using the pattern for}%
\typeout{** the default language instead.}%
\else
\language=\csname l@#1\endcsname
\fi
#2}}
\providecommand{\BIBdecl}{\relax}
\BIBdecl

\bibitem{Atzori.2017}
L.~Atzori, A.~Iera, and G.~Morabito, ``Understanding the {Internet of Things}:
  definition, potentials, and societal role of a fast evolving paradigm,''
  \emph{Ad Hoc Networks}, vol.~56, pp. 122--140, 2017.

\bibitem{Chen.2014}
K.-C. Chen and S.-Y. Lien, ``Machine-to-machine communications: Technologies
  and challenges,'' \emph{Ad Hoc Networks}, vol.~18, pp. 3--23, 2014.

\bibitem{Ericsson.2017}
\BIBentryALTinterwordspacing
Ericsson. (2017, nov) Ericsson mobility report. [Online]. Available:
  \url{https://www.ericsson.com/en/mobility-report/reports/november-2017/internet-of-things-outlook}
\BIBentrySTDinterwordspacing

\bibitem{Dawy.2017}
Z.~Dawy, W.~Saad, A.~Ghosh, J.~G. Andrews, and E.~Yaacoub, ``Toward massive
  machine type cellular communications,'' \emph{IEEE Wireless Commun.},
  vol.~24, no.~1, pp. 120--128, February 2017.

\bibitem{Rajandekar.2015}
A.~Rajandekar and B.~Sikdar, ``A survey of {MAC} layer issues and protocols for
  machine-to-machine communications,'' \emph{IEEE Internet of Things Journal},
  vol.~2, no.~2, pp. 175--186, April 2015.

\bibitem{ETSI1}
{3GPP TS 36.331 V10.50.0}, ``Evolved universal terrestrial radio access
  ({E-UTRA}); radio resource control ({RRC}),'' Mar 2012.

\bibitem{Lin.2014}
T.~M. Lin, C.~H. Lee, J.~P. Cheng, and W.~T. Chen, ``{PRADA}: Prioritized
  random access with dynamic access barring for {MTC} in {3GPP LTE-A
  }networks,'' \emph{IEEE Trans. on Vehic. Tech.}, vol.~63, no.~5, pp.
  2467--2472, Jun 2014.

\bibitem{Yang.2012}
X.~Yang, A.~Fapojuwo, and E.~Egbogah, ``Performance analysis and parameter
  optimization of random access backoff algorithm in {LTE},'' in \emph{IEEE
  VTC-Fall}, Sept 2012, pp. 1--5.

\bibitem{Hossain.2016}
M.~I. Hossain, A.~Azari, and J.~Zander, ``{DERA}: {Augmented} random access for
  cellular networks with dense {H2H-MTC} mixed traffic,'' in \emph{2016 IEEE
  Globecom Workshops (GC Wkshps)}, Dec 2016, pp. 1--7.

\bibitem{Laya.2016}
A.~Laya, C.~Kalalas, F.~Vazquez-Gallego, L.~Alonso, and J.~Alonso-Zarate,
  ``Goodbye, {ALOHA}!'' \emph{IEEE Access}, vol.~4, pp. 2029--2044, 2016.

\bibitem{Shariatmadari.2015}
H.~Shariatmadari, P.~Osti, S.~Iraji, and R.~J\"antti, ``Data aggregation in
  capillary networks for machine-to-machine communications,'' in \emph{2015
  IEEE 26th Annual International Symposium on Personal, Indoor, and Mobile
  Radio Communications (PIMRC)}, Aug 2015, pp. 2277--2282.

\bibitem{nardelli2016maximizing}
P.~H. Nardelli, M.~de~Castro~Tom{\'e}, H.~Alves, C.~H. de~Lima, and
  M.~Latva-aho, ``Maximizing the link throughput between smart meters and
  aggregators as secondary users under power and outage constraints,'' \emph{Ad
  Hoc Networks}, vol.~41, pp. 57--68, 2016.

\bibitem{Ramezanipour.2018}
I.~{Ramezanipour}, P.~{Nouri}, H.~{Alves}, P.~H.~J. {Nardelli}, R.~D. {Souza},
  and A.~{Pouttu}, ``Finite blocklength communications in smart grids for
  dynamic spectrum access and locally licensed scenarios,'' \emph{IEEE Sensors
  Journal}, vol.~18, no.~13, pp. 5610--5621, July 2018.

\bibitem{Boubiche.2018}
S.~Boubiche, D.~E. Boubiche, A.~Bilami, and H.~Toral-Cruz, ``Big data
  challenges and data aggregation strategies in wireless sensor networks,''
  \emph{IEEE Access}, vol.~6, pp. 20\,558--20\,571, 2018.

\bibitem{Kouzayha.2014}
N.~Kouzayha, M.~Jaber, and Z.~Dawy, ``{M2M} data aggregation over cellular
  networks: {Signaling-delay} trade-offs,'' in \emph{2014 IEEE Globecom
  Workshops (GC Wkshps)}, Dec 2014, pp. 1155--1160.

\bibitem{Riker.2014}
A.~Riker, E.~Cerqueira, M.~Curado, and E.~Monteiro, ``Data aggregation for
  group communication in machine-to-machine environments,'' in \emph{2014 IFIP
  Wireless Days (WD)}, Nov 2014, pp. 1--7.

\bibitem{Bormann.2012}
C.~Bormann, A.~P. Castellani, and Z.~Shelby, ``{CoAP}: {An} application
  protocol for billions of tiny {Internet} nodes,'' \emph{IEEE Internet
  Computing}, vol.~16, no.~2, pp. 62--67, March 2012.

\bibitem{Fitzgerald.2018}
E.~Fitzgerald, M.~Pi\'oro, and A.~Tomaszewski, ``Energy-optimal data
  aggregation and dissemination for the {Internet of Things},'' \emph{IEEE
  Internet of Things Journal}, vol.~5, no.~2, pp. 955--969, April 2018.

\bibitem{Guo.2017}
J.~Guo, S.~Durrani, X.~Zhou, and H.~Yanikomeroglu, ``Massive machine type
  communication with data aggregation and resource scheduling,'' \emph{IEEE
  Transactions on Communications}, vol.~65, no.~9, pp. 4012--4026, Sept 2017.

\bibitem{Shirvanimoghaddam.2016}
M.~Shirvanimoghaddam, M.~Dohler, and S.~J. Johnson, ``Massive non-orthogonal
  multiple access for cellular {IoT}: {Potentials} and limitations,''
  \emph{IEEE Communications Magazine}, vol.~55, no.~9, pp. 55--61, 2017.

\bibitem{Saito.2013_2}
Y.~Saito, A.~Benjebbour, Y.~Kishiyama, and T.~Nakamura, ``System-level
  performance evaluation of downlink non-orthogonal multiple access {(NOMA)},''
  in \emph{IEEE PIMRC}, Sept 2013, pp. 611--615.

\bibitem{Ding.2014}
Z.~Ding, Z.~Yang, P.~Fan, and H.~V. Poor, ``On the performance of
  non-orthogonal multiple access in {5G} systems with randomly deployed
  users,'' \emph{IEEE Sig. Proces. Lett.}, vol.~21, no.~12, Dec 2014.

\bibitem{Zhang.2016}
Z.~Zhang, H.~Sun, R.~Q. Hu, and Y.~Qian, ``Stochastic geometry based
  performance study on {5G} non-orthogonal multiple access scheme,'' in
  \emph{IEEE GLOBECOM}, Dec 2016, pp. 1--6.

\bibitem{Lopez.2017}
O.~L.~A. L\'opez, H.~Alves, P.~H.~J. Nardelli, and M.~Latva-aho, ``Aggregation
  and resource scheduling in machine-type communication networks: {A}
  stochastic geometry approach,'' \emph{IEEE Transactions on Wireless
  Communications}, vol.~17, no.~7, pp. 4750--4765, July 2018.

\bibitem{Liu.2016}
F.~Liu, P.~M\"ah\"onen, and M.~Petrova, ``Proportional fairness-based power
  allocation and user set selection for downlink {NOMA} systems,'' in
  \emph{2016 IEEE International Conference on Communications (ICC)}, May 2016,
  pp. 1--6.

\bibitem{Shirvanimoghaddam.2017}
M.~Shirvanimoghaddam, M.~Dohler, and S.~J. Johnson, ``Massive multiple access
  based on superposition raptor codes for cellular {M2M} communications,''
  \emph{IEEE Trans. on Wireless Commun.}, vol.~16, no.~1, pp. 307--319, Jan
  2017.

\bibitem{Malak.2016}
D.~Malak, H.~S. Dhillon, and J.~G. Andrews, ``Optimizing data aggregation for
  uplink machine-to-machine communication networks,'' \emph{IEEE Trans. on
  Commun.}, vol.~64, no.~3, pp. 1274--1290, March 2016.

\bibitem{Haenggi.2012}
M.~Haenggi, \emph{Stochastic geometry for wireless networks}.\hskip 1em plus
  0.5em minus 0.4em\relax Cambridge University Press, 2012.

\bibitem{Gharbieh.2017}
M.~Gharbieh, H.~ElSawy, A.~Bader, and M.~S. Alouini, ``Spatiotemporal
  stochastic modeling of {IoT} enabled cellular networks: Scalability and
  stability analysis,'' \emph{IEEE Transactions on Communications}, vol.~65,
  no.~8, pp. 3585--3600, Aug 2017.

\bibitem{Haenggi.2014}
M.~Haenggi, ``The mean interference-to-signal ratio and its key role in
  cellular and amorphous networks,'' \emph{IEEE Wireless Commun. Letters},
  vol.~3, no.~6, pp. 597--600, Dec 2014.

\bibitem{Ganti.2016}
R.~K. Ganti and M.~Haenggi, ``Asymptotics and approximation of the {SIR}
  distribution in general cellular networks,'' \emph{IEEE Trans. on Wireless
  Commun.}, vol.~15, no.~3, pp. 2130--2143, March 2016.

\bibitem{Sun.2016}
H.~Sun, B.~Xie, R.~Q. Hu, and G.~Wu, ``Non-orthogonal multiple access with
  {SIC} error propagation in downlink wireless {MIMO} networks,'' in \emph{IEEE
  Vehicular Technology Conference (VTC-Fall)}, Sept 2016, pp. 1--5.

\end{thebibliography}
\end{document}